\documentclass[11pt,fleqn]{article}

\title{Constant approximation for fault-tolerant median problems\\via iterative rounding}

\author{Shichuan Deng}

\date{Tsinghua University, China}

\usepackage{amsthm}
\usepackage{amsmath}
\usepackage{amsfonts,bbm}
\usepackage{amssymb}
\usepackage{mathrsfs}
\usepackage{graphicx,url}
\usepackage{float}
\usepackage[margin=1in]{geometry}
\usepackage[dvipsnames]{xcolor}
\usepackage[linesnumbered,ruled,vlined,noresetcount]{algorithm2e}
\usepackage[figurename=Fig.,labelfont=bf,labelsep=period]{caption}
\usepackage{xspace}
\usepackage{enumitem}
\usepackage{hyperref}
\usepackage[dvipsnames]{xcolor}
\usepackage[capitalise]{cleveref}
\hypersetup{
  colorlinks=true,
  linkcolor=black,
  urlcolor=cyan,
  citecolor=olive,
}

\newcommand{\floor}[1]{\left\lfloor #1\right\rfloor}

\newcommand{\etal}{\textit{et~al.}\xspace}
\newcommand{\etalcite}[1]{\textit{et~al.}~\cite{#1}}

\newcommand{\R}{\mathbb{R}}
\newcommand{\Z}{\mathbb{Z}}
\newcommand{\calB}{\mathcal{B}}
\newcommand{\calI}{\mathcal{I}}

\newcommand{\calC}{\mathcal{C}}
\newcommand{\calF}{\mathcal{F}}
\newcommand{\calM}{\mathcal{M}}
\newcommand{\calU}{\mathcal{U}}
\newcommand{\calX}{\mathcal{X}}

\newcommand{\ball}{\textsc{\small Ball}}
\newcommand{\ib}{B}
\newcommand{\rg}{{\gamma^{-1}}}
\newcommand{\dmd}{n}
\newcommand{\que}{{\mathscr{Q}}}

\newcommand{\opt}{\mathsf{OPT}}
\newcommand{\dav}{d_{\mathsf{av}}}
\newcommand{\dma}{d_{\mathsf{max}}}
\newcommand{\ftotal}{f_{\mathsf{total}}}

\newcommand{\ftfl}{\textsf{FTFL}\xspace}
\newcommand{\mmd}{\textsf{MtMed}\xspace}
\newcommand{\kmd}{\textsf{KpMed}\xspace}
\newcommand{\ftm}{\textsf{FTMtMed}\xspace}
\newcommand{\ftk}{\textsf{FTKpMed}\xspace}

\newtheorem{theorem}{Theorem}
\newtheorem{lemma}[theorem]{Lemma}

\theoremstyle{definition}

\theoremstyle{remark}
\newtheorem{remark}{Remark}

\begin{document}
\maketitle

\begin{abstract}
    In this paper, we study the \emph{fault-tolerant matroid median} and \emph{fault-tolerant knapsack median} problems. 
    These two problems generalize many fundamental clustering and facility location problems, such as uniform fault-tolerant $k$-median, uniform fault-tolerant facility location, matroid median, knapsack median, etc.
    We present a versatile iterative rounding framework and obtain a unifying constant-factor approximation algorithm.
\end{abstract}

\section{Introduction}

We study matroid median and knapsack median problems under the notion of fault-tolerance. 
In both matroid median (\mmd) and knapsack median (\kmd), we are given clients $\calC$, facilities $\calF$, a finite metric $d$ on $\calC\cup\calF$, and the goal is to select a subset $F\subseteq\calF$ so as to minimize $\sum_{j\in\calC}d(j,F)$, where $d(j,F):=\min_{i\in F}d(j,i)$. 
In \mmd, we require $F$ to be an independent set of a given matroid $\calM$. 
In \kmd, each facility $i\in\calF$ has a weight $w_i\geq0$ and the selected facilities $F$ should have a total weight no more than a given threshold $W$; that is, $\sum_{i\in F}w_i\leq W$.
\mmd and \kmd generalize the classic $k$-median problem (e.g., \cite{charikar2012dependent}). 
The best approximation algorithms have factors $7.081$ \cite{krishnaswamy2018constant} and $(6.387+\epsilon)$ \cite{gupta2020structural}, respectively.

Formally, in fault-tolerant matroid median (\ftm), we are given a finite metric space $(\calC\cup\calF,d)$, a matroid $\calM=(\calF,\calI)$ and a \emph{requirement} $r\in\mathbb{Z}_+$. Each facility $i\in\calF$ has an \emph{opening cost} $f_i\geq0$. 
We need to open facilities in an independent set $F\in\calI$ and assign the nearest $r$ open facilities to each client. The service cost incurred for client $j$ is the sum of distances from $j$ to its assigned facilities, and the goal is to minimize the sum of facility opening costs and client service costs, i.e., $
\sum_{i\in F}f_i+\sum_{j\in\calC}d_r(j,F)$, where $d_r(j,F):=\min_{F'\subseteq F:|F'|=r}\sum_{i\in F'}d(i,j)$. 
This is a generalization of \mmd, which is the setting where $r=1$ and each $f_i=0$.
Similarly, we define the fault-tolerant knapsack median problem (\ftk) based on \kmd and fault-tolerant assignments of open facilities.

In fact, one may consider a more general formulation of fault-tolerance, such that each client $j$ has a \emph{distinct} requirement $r_j\in\mathbb{Z}_+$ governing the number of open facilities assigned to it in a feasible solution. We call these formulations \emph{non-uniform} for clarity, and this paper is only concerned with \emph{uniform} requirements.
Fault-tolerant facility location (\ftfl) is first introduced by Jain and Vazirani \cite{jain2004approximation} as a natural generalization of uncapacitated facility location. The current best approximation ratios are 1.725 \cite{byrka2010fault} and 1.52 \cite{swamy2008fault} for non-uniform and uniform \ftfl, respectively.
Hajiaghayi~\etalcite{hajiaghayi2016constant} consider fault-tolerant $k$-median with non-uniform requirements and give a 93-approximation.

Clustering problems with stronger combinatorial constraints have recently received more research interest. 
Hochbaum and Shmoys \cite{hochbaum1986unified} study knapsack center and give a 3-approximation.
Chen~\etal\cite{chen2016matroid} devise a 3-approximation for matroid center. 
Kumar \cite{kumar2012constant} gives the first constant-factor approximation for knapsack median, which is later improved in \cite{byrka2015knapsack,charikar2012dependent,gupta2020structural,krishnaswamy2018constant,swamy2016improved}. 
Krishnaswamy~\etal\cite{krishnaswamy2011matroid} gives a constant-factor approximation for matroid median, which is later improved in \cite{charikar2012dependent,krishnaswamy2018constant,swamy2016improved}.

\section{Fault-tolerant matroid median}

We present the natural LP relaxation for \ftm, where $x_{ij}\in[0,1]$ represents the extent of assignment between facility $i$ and client $j$, and $y_i\in[0,1]$ represents the extent we open facility $i$.
Let $r_\calM:2^\calF\rightarrow\mathbb{Z}$ be the rank function of matroid $\calM$.
\begin{alignat*}{1}
    \text{min\quad} &
    \sum_{i\in\calF}f_iy_i +\sum_{j\in\calC}\sum_{i\in\calF}x_{ij}d(i,j)\tag{$\operatorname{M-LP}$}\label{lp:natural:matroid}\\
    \text{s.t.\quad}
    & \sum_{i\in\calF}x_{ij} = r\quad\forall j\in\calC\\
    & 0\leq x_{ij} \leq y_i\leq1\quad\forall i\in\calF,j\in\calC\\
    & \sum_{i\in S}y_i \leq r_\calM(S)\quad\forall S\subseteq\calF.
\end{alignat*}
Using a classic result by Edmonds \cite{edmonds2001submodular}, \ref{lp:natural:matroid} can be solved efficiently using the ellipsoid method.
Fix an optimal solution $(x,y)$ to \ref{lp:natural:matroid} in what follows. 
By splitting each facility into many copies if necessary \cite{charikar2012dependent} and always letting $\calF'$ denote the set of facilities after duplication, we assume $x_{ij}\in\{0,y_i\}$ for each $i\in\calF',j\in\calC$.
We maintain another matroid $\calM'$ on $\calF'$, such that for any $S'\subseteq\calF'$ and $S\subseteq\calF$ where $S'$ consists of copies of $S$, we have $r_\calM(S)=r_{\calM'}(S')$. 
Throughout this paper, for each $x\in\R^\calX$ and $S\subseteq\calX$, we write $x(S)=\sum_{i\in S}x_i$.


\subparagraph{Solution-dependent notation.}
Let $\calF_j=\{i\in\calF':x_{ij}>0\}$, and $y(\calF_j)=r$ by definition. 
Since $x_{ij}=y_i$ for each $i\in\calF_j$, the $y$-vector on $\calF'$ and subsets $\{\calF_j\subseteq\calF':j\in\calC\}$ can fully represent the solution $(x,y)$. 
By making co-located copies, there exists a partition $\calF_j=\bigcup_{t\in[r]}\calF_{j,t}$ for each $j$ such that $y(\calF_{j,t})=1$ for each $t\in[r]$, and $\calF_{j,t}$ is the $t$-th nearest unit volume of facilities in $\calF_j$; 
that is, for $t_1<t_2$ and $i_1\in\calF_{j,t_1},i_2\in\calF_{j,t_2}$, one has $d(j,i_1)\leq d(j,i_2)$.
Define $\dav^t(j)=\sum_{i\in\calF_{j,t}}y_id(i,j)$ the average distance from $j$ to $\calF_{j,t}$, and $\dma^t(j)=\max_{i\in\calF_{j,t}}d(i,j)$. 
The following inequalities then easily follow,
\begin{equation}
\dav^1(j)\leq\dma^1(j)\leq\dav^2(j)\leq\cdots\leq\dav^{r}(j)\leq\dma^{r}(j).
\label{equation:chain:of:inequalities}
\end{equation}

We define $\dav(j)=r^{-1}\sum_{i\in\calF_j}y_id(i,j)$, the total opening cost $\ftotal=\sum_{i\in\calF'}f_iy_i$ and write $\dma(j)=\dma^{r}(j)$ as a shorthand term. It is easy to see $r\dav(j)$ is the contribution of $j$ to \ref{lp:natural:matroid}, and $r\dav(j)=\sum_{t=1}^{r}\dav^t(j)$. We reload $\dma(j,S)=\max_{i\in S}d(i,j)$ to represent the maximum distance from $j$ to any $S\subseteq\calF'$. 

For $j\in\calC$ and each $R,\lambda\geq0$, we define the closed balls of facilities $\ball(j,R)=\{i\in\calF':d(j,i)\leq R\}$ and $\ball_\lambda(j)=\ball(j,\lambda\cdot\dma(j))$.
Clearly, since $(x,y)$ is optimal, $\dma(j)$ is the smallest radius $R$ such that $y(\ball(j,R))\geq r$ for each $j\in\calC$.

\section{The iterative rounding framework}\label{section:matroid:uniform}

We introduce an iterative rounding framework with several useful ingredients adapted from \cite{hajiaghayi2016constant}, and devise a constant-factor approximation for \ftm. 
Our main goal is to round an auxiliary LP that has integral vertex solutions.
Previously for fault-tolerant $k$-median, Hajiaghayi \etal\cite{hajiaghayi2016constant} develop some linear constraints to restrict the placement of open facilities.
In particular, they obtain an auxiliary LP defined by two laminar families, which is well-known to be integral since the coefficient matrix is totally unimodular (see, e.g., \cite{schrijver2003combinatorial}).

In \ftm, however, the auxiliary LP is not necessarily integral because of the matroid $\calM$.
We overcome this issue by simplifying the laminar families into two families of disjoint subsets denoted by $\calB$ and $\calU$, and iteratively modifying $\calU$ such that it eventually becomes a ``refinement'' of $\calB$; that is, the two families still contain disjoint subsets, and each $B\in\calB$ becomes a union of some subsets in $\calU$.
As a result, the LP constraints defined by $\calB$ become obsolete and other constraints are defined by $\calM$ and $\calU$ (a partition matroid), hence the auxiliary LP becomes integral (see \cite{edmonds2001submodular}).

\subsection{Construction of the family \texorpdfstring{$\calB$}{}}

In this section, we create a family $\calB$ of disjoint closed balls, each of which is centered at a \emph{dangerous} client, formally defined as follows: Fix $\gamma=3+\delta$ where $\delta>0$, and define the set of dangerous clients as $D=\{j\in\calC:\dma(j)>3\gamma\dav^r(j)\}$; if a client is not dangerous, it is \emph{safe}.

We say that two dangerous clients $j$ and $j'$ are \emph{in conflict} if $d(j,j')\leq6\max\{\dav(j),\dav(j')\}$ and construct a filtered subset $D'\subseteq D$ as follows.
Set $D'\leftarrow\emptyset$ initially. 
In non-decreasing order of $\dav(j)$ for each \emph{unmarked} $j\in D$, $D'\leftarrow D'\cup\{j\}$ and \emph{mark} each $j'\in D$ in conflict with $j$. 
Put $\dmd_j\in\mathbb{Z}_+$ as the number of \emph{newly-marked} clients in this step (including $j$ itself) and consolidate a demand of $\dmd_j$ at $j$. 

Let $\ib_j=\ball_{1/\gamma}(j)=\{i\in\calF':d(i,j)\leq\dma(j)/\gamma\}$ for each $j\in D'$, and $\calB=\{\ib_j:j\in D'\}$.
We note that $\calB$ is a simplified version of the laminar family in \cite{hajiaghayi2016constant} and defined using different radii.
The purpose of the closed balls is to create certain LP constraints (see \eqref{lp:iterative3}) that prioritize opening facilities inside each $B_j\in\calB$, because the $r$-th partition $\calF_{j,r}$ of $j\in D'$ is more ``diffusive'' than safe clients in the sense of $\dma^r(j)>3\gamma\cdot\dav^r(j)$, and bounding the cost of dangerous clients in a rounded integral solution is much more challenging.
Since $\gamma>3$, the following lemma shows that the closed balls in $\calB$ are pair-wise disjoint.

\begin{lemma}\label{lemma:first:separation}
For $j\neq j'$ from $D'$, one has \[d(j,j')\geq\max\{\dma(j),\dma(j')\}-\rg\cdot\min\{\dma(j),\dma(j')\}.\]
\end{lemma}
\begin{proof}

We first claim $y(\ib_j)\in[r-1/3,r)$ for each $j\in D'$ \cite{hajiaghayi2016constant}.
Suppose $y(\ib_j)<r-1/3$, then $y(\calF_{j,r}\setminus \ib_j)>1/3$ and the average distance on $\calF_{j,r}$ is at least $\dav^r(j)\geq\dma(j)/(3\gamma)>\dav^r(j)$ since $j$ is dangerous, which is a contradiction.

Next, we show $d(j,j')\geq\dma(j)/3+\dma(j')/3$. Assume otherwise and w.l.o.g., $\dma(j)\geq\dma(j')$. We thus have $d(j,j')+\dma(j')/\gamma\leq 2\dma(j)/3+\dma(j)/\gamma$, so the closed ball $\ib_{j'}$ satisfies $\ib_{j'}=\ball_{1/\gamma}(j')\subseteq\ball_{2/3+\rg}(j)$ and $y(\ib_j\cup \ib_{j'})\leq y(\ball_{2/3+\rg}(j))<r$ since $2/3+\rg<1$ and $\ball_{2/3+\rg}(j)\subsetneq\calF_j$. This gives
$y(\ib_j\cap \ib_{j'})=y(\ib_j)+y(\ib_{j'})-y(\ib_j\cup \ib_{j'})\geq r-2/3$.

Using the triangle inequality, if we sample $i\in \ib_j\cap \ib_{j'}$ on the normalized distribution defined by $y$, it follows that $d(j,j')\leq\mathbb{E}_{i}[d(j,i)+d(i,j')]=\mathbb{E}_i[d(i,j)]+\mathbb{E}_i[d(i,j')]$. 
Since $y(\ib_j\cap \ib_{j'})\geq r-2/3$, this is at most $r\dav(j)/(r-2/3)+r\dav(j')/(r-2/3)\leq 3\dav(j)+3\dav(j')\leq 6\max\{\dav(j),\dav(j')\}$, contradicting the fact that $j$ and $j'$ are not in conflict.

Now $\ib_j\cap \ib_{j'}=\emptyset$ using $1/\gamma<1/3$. 
Assume the lemma is false and w.l.o.g., $\dma(j)\geq\dma(j')$. 
Let $R=d(j,j')+\dma(j')/\gamma<\dma(j)$, thus $\ib_{j'}\subseteq\ball(j,R)$.
Since $\ib_j\cap \ib_{j'}=\emptyset$, we have $y(\ball(j,R))\geq y(\ib_j)+y(\ib_{j'})\geq 2r-2/3>r$ for $R<\dma(j)$, contradicting the definition of $\dma(j)$.
\end{proof}


\subsection{Construction of the family \texorpdfstring{$\calU$}{}}

We create a disjoint family $\calU$ of subsets of $\calF'$ called \emph{bundles} via the adaptive clustering algorithm by Yan and Chrobak \cite{yan2015lp}, with each of them having volume 1; 
i.e., $\forall U\in\calU$, $y(U)=1$.
We will eventually open exactly one facility in each bundle.
We also assign a set of distinct bundles $\que_j\subseteq\calU$ to each $j\in\calC$, called the \emph{queue} of $j$.
Let $Q_{j,t}$ be the $t$-th bundle added to $\que_j$, $Q_{j,<l}=\bigcup_{t<l}Q_{j,t}$, and $Q_{j,\leq l}=\bigcup_{t\leq l}Q_{j,t}$.
The algorithm is described in \cref{algorithm:bundle}; it is also where our procedure diverges from \cite{hajiaghayi2016constant,yan2015lp}, by creating the bundles and queues using a more fine-tuned approach.
At \cref{line:dangerous2} and \cref{line:safe:creation}, we say that client $j$ is the \emph{creator} of bundle $U$.

Recall our goal is to make $\calU$ into a refinement of $\calB$, thus in the initial construction, we reduce the extent of intersecting between their members (we say $S_1$ and $S_2$ are \emph{intersecting} if $S_1\setminus S_2,S_2\setminus S_1,S_1\cap S_2$ are all non-empty).
Each $j\in D'$ has $|\que_j|=r$, because whenever a bundle is added to $\que_j$ (\cref{line:dangerous1} and \cref{line:dangerous2}), $y(\calF_j')$ which starts from $r$ is decreased by at most 1. 
Each $j''\in D\setminus D'$ has $|\que_{j''}|=0$, which simplifies the structure of $\calU$; there exists $j\in D'$ in conflict with $j''$, which can always provide the assigned $r$ facilities in the final solution.


$\calU_{shell}\subseteq\calU$ is a subset of \emph{shell bundles}, each of which is created by some dangerous $j\in D'$ when it happens to be the $r$-th bundle in $\que_j$ (\cref{line:create:shell}).
As we show later in \cref{lemma:queue:and:bundle}, we want $\calU_{shell}$ to be the \emph{only} bundles that could be intersecting the closed balls in $\calB$; 
to complete the refinement, we will remove some shell bundles during iterative rounding. 
This also explains \cref{line:noalien} and \cref{line:noshell} where one freezes a safe queue early by setting $\calF_j'\leftarrow\emptyset$: 
when a safe client $j$ wants to create a bundle intersecting some $B_{j'}\in\calB$, or when a shell bundle created by $j'\in D'$ is going to be added to a safe queue $\que_{j}$, we freeze $\que_{j}$ and ``borrow'' some facilities (see \cref{lemma:final:solution:safe}) from $j'$ for $j$ in the final solution.
Therefore, each $j\in\calC\setminus D$ has $|\que_{j}|\leq r$.

\begin{algorithm}[hbt!]
\caption{\textsc{ALG-Bundle}}\label{algorithm:bundle}
\DontPrintSemicolon
\SetKwInOut{Input}{Input}
\SetKwInOut{Output}{Output}

\Input{$\calF',\calC,D,D';\,(x,y),\{\calF_j:j\in\calC\},\{\ib_j:j\in D'\}$}
\Output{bundles and subsets of bundles for each client}
\tcc{duplicate facilities in $\calF'$ whenever necessary}
$\calU\leftarrow\emptyset,\,\calU_{shell}\leftarrow\emptyset;\,\forall j\in\calC,\,\calF_j'\leftarrow \calF_j,\,\que_j\leftarrow\emptyset$\;
\While{$\exists j\in(\calC\setminus D)\cup D'$ s.t. $|\que_j|<r$ and $\calF_j'\neq\emptyset$}{
choose such $j$, s.t. for $U$ being the nearest unit volume of facilities in $\calF_j'$ (i.e., $y(U)=1$), $\max_{i\in U}d(i,j)$ is minimized\label{line:propose}\;
\uIf(\tcc*[f]{$j$ is dangerous}){$j\in D'$}{
\lIf{\label{line:dangerous1}$\exists U'\in\calU$ s.t. $U\cap U'\neq\emptyset$}{
$\que_j\leftarrow\que_j\cup\{U'\}$, $\calF_j'\leftarrow\calF_j'\setminus U'$
}
\Else(\label{line:dangerous2}$\que_j\leftarrow\que_j\cup\{U\},\,\calU\leftarrow\calU\cup\{U\}$, $\calF_j'\leftarrow\calF_j'\setminus U$\tcc*[f]{$U$ is created by $j$}){
\lIf{$|\que_j|=r$}{$\calU_{shell}\leftarrow\calU_{shell}\cup\{U\}$\label{line:create:shell}}
}
}
\Else(\tcc*[f]{$j$ is safe}){
\lIf{$\exists j'\in D'$ s.t. $U\cap \ib_{j'}\neq\emptyset,\,U\setminus \ib_{j'}\neq\emptyset$, $U$ disjoint from $\que_{j'}$\label{line:noalien}}{
$\calF_j'\leftarrow\emptyset$}
\lElseIf{$\exists U'\in\calU_{shell}$ s.t. $U\cap U'\neq\emptyset$\label{line:noshell}}{$\calF_j'\leftarrow\emptyset$}
\lElseIf{$\exists U'\in\calU$ s.t. $U\cap U'\neq\emptyset$}{
$\que_j\leftarrow \que_j\cup\{U'\}$, $\calF_j'\leftarrow\calF_j'\setminus U'$
}
\lElse{\label{line:safe:creation}$\que_j\leftarrow \que_j\cup\{U\},\,\calU\leftarrow\calU\cup\{U\}$, $\calF_j'\leftarrow\calF_j'\setminus U$\tcc*[f]{$U$ is created by $j$}}
}
}
\Return $\calU,\,\calU_{shell}\subseteq\calU,\,\{\que_j\subseteq\calU:j\in\calC\}$
\end{algorithm}

The following lemma follows directly from \cite{hajiaghayi2016constant}.

\begin{lemma}\label{lemma:bundle:third:away}\emph{(\cite{hajiaghayi2016constant}).}
$\dma(j,Q_{j,t})\leq3\dma^t(j)$, $\forall j\in\calC$,  $t\in[r]$.
\end{lemma}

Next, we establish two lemmas on some special properties of bundles and queues, as a result of respecting the dangerous clients in \cref{algorithm:bundle}.
Especially, \cref{lemma:queue:and:bundle} shows that $\calU$ is already close to being a refinement of $\calB$ after \cref{algorithm:bundle}.

\begin{lemma}\label{lemma:noalien:safe}
At the moment when \cref{line:noalien} of \cref{algorithm:bundle} holds, one has $|\que_{j'}|\geq r-1$ and $\dma(j,U)\geq(1-\rg)\dma(j')/2$.
\end{lemma}
\begin{proof}
We restate the condition for convenience: A safe client $j$ proposes a candidate bundle $U\subseteq\calF_j'$. There exists $j'\in D'$ s.t. $U\cap \ib_{j'}\neq\emptyset,\,U\setminus \ib_{j'}\neq\emptyset$ and $U$ is disjoint from $\que_{j'}$.

Assume otherwise and $|\que_{j'}|\leq r-2$, thus if $j'$ were to propose its candidate bundle $U'$ in this iteration, it would have $\dma(j',U')\leq\dma^{r-1}(j')$ since at most $(r-2)$ volume of facilities is removed from $\calF'_{j'}$. 
Thus $\dma(j,U)\leq\dma^{r-1}(j')$ since $U$ is proposed before $U'$. 
We claim $\min_{i\in U}d(i,j')>\dma^{r-1}(j')$, otherwise $\dma(j',U)\leq\dma^{r-1}(j')+2\dma(j,U)\leq3\dma^{r-1}(j')\leq3\dav^r(j')<\dma(j')/\gamma$ using \cref{equation:chain:of:inequalities}, hence $U\subseteq B_{j'}$ which contradicts $U\setminus \ib_{j'}\neq\emptyset$.
This implies $\ball(j',\dma^{r-1}(j'))\cap U=\emptyset$.

Since $U\cap\ib_{j'}\neq\emptyset$, we have $\dma(j',U)\leq\dma(j')/\gamma+2\dma^{r-1}(j')\leq\dma(j')/\gamma+2\dav^r(j')<\dma(j')(1/\gamma+2/(3\gamma))<0.9\dma(j')$ using the triangle inequality.
Thus, one further has $U\subseteq\ball_{0.9}(j')$ and $y(\ball_{0.9}(j'))\geq y(\ball(j',\dma^{r-1}(j')))+y(U)\geq r$ because $U$ and $\ball(j',\dma^{r-1}(j'))$ are disjoint due to the analysis above.
This contradicts the definition of $\dma(j')$.

Therefore, we must have $|\que_{j'}|\geq r-1$.
For the sake of contradiction, we assume $\dma(j,U)<(1-\rg)\dma(j')/2$. 
Therefore, $\dma(j',U)\leq \dma(j')/\gamma+2\dma(j,U)<\dma(j')$ using the triangle inequality, $\dma(j',Q_{j',<r})\leq3\dma^{r-1}(j')\leq3\dav^r(j')<\dma(j')$ using \cref{lemma:bundle:third:away}, and $U$ is disjoint from $\que_{j'}$ by the condition.
As a result, there is at least $r$ volume of facilities $Q_{j',<r}\cup U$ inside a ball centered at $j'$ with a radius smaller than $\dma(j')$, contradicting the definition of $\dma(j')$ again.
\end{proof}

\begin{lemma}\label{lemma:queue:and:bundle}
(i) For $j'\in D'$, one has $Q_{j',<r}\subseteq \ib_{j'}$.
(ii) For $U\in\calU$ and $j'\in D'$, $U\subseteq\ib_{j'}$ if and only if $U\subseteq Q_{j',<r}$.
(iii) For $j\in\calC\setminus D$ and $j'\in D'$, each $Q_{j,t}$ satisfies either $Q_{j,t}\cap \ib_{j'}=\emptyset$ or $Q_{j,t}\subseteq\ib_{j'}$.
\end{lemma}
\begin{proof}
For each $j'\in D'$ and $t<r$, using \cref{lemma:bundle:third:away} we have $\dma(j',Q_{j',t})\leq3\dma^{t}(j')\leq3\dav^{t+1}(j')<\dma(j')/\gamma$, thus (i) follows.
For (ii), the if-direction follows from (i); for the only-if-direction, since $y(\ib_{j'})\in[r-1/3,r)$ and the bundles are mutually disjoint, $U$ has to be one of the bundles in $Q_{j',<r}$, otherwise there would be at least $r$ volume of facilities $Q_{j',<r}\cup U$ in $\ib_{j'}$, which contradicts $y(\ib_{j'})<r$.

For (iii), assume otherwise and $Q_{j,t}\cap \ib_{j'}\neq\emptyset,\,Q_{j,t}\nsubseteq \ib_{j'}$. 
First, $Q_{j,t}$ is not created by any dangerous $j''\in D'$, since such a bundle is either completely inside $\ib_{j''}$ (the first $(r-1)$ bundles in $\que_{j''}$) or a shell bundle (the last bundle in $\que_{j''}$), which cannot be in the safe queue $\que_{j}$ due to \cref{line:noshell}. 
Thus, $Q_{j,t}$ is created by a safe client. But such a bundle would satisfy \cref{line:noalien} when it is created and should be discarded, a contradiction.
\end{proof}


\subsection{The auxiliary LP and iterative rounding}

We use an auxiliary LP and iterative rounding (\cref{algorithm:iterative}) to obtain an integral solution $z^\star\in\{0,1\}^{\calF'}$, where $z_i\in[0,1]$ is the extent we open $i\in\calF'$.
We maintain two subsets $D_0,D_1\subseteq D'$, which are initially empty. 
Each $j\in D_0$ will have exactly $(r-1)$ bundles constituting $\ib_j$, and each $j\in D_1$ will have $r$ ones constituting $\ib_j$.
We partition $D'=D_0\cup D_1$ in the end by adding each $j$ to $D_0$ or $D_1$. 
Recall when creating $D'\subseteq D$, $\dmd_j\in\mathbb{Z}_+$ is the number of dangerous clients newly-marked by and in conflict with $j\in D'$.

\begin{alignat}{1}
    \text{min\;}&\sum_{j\in D'\setminus(D_0\cup D_1)}\dmd_j\left(\left(\sum_{i\in \ib_j}z_i d(i,j)\right) +\frac{\dma(j)}{\gamma}(r-z(\ib_j)) \right)+\notag\\
    &\sum_{j\in D_0}\dmd_j\sum_{i\in Q_{j,<r}}z_i d(i,j)+
    \sum_{j\in D_1}\dmd_j\sum_{i\in Q_{j,\leq r}}z_i d(i,j)+\sum_{i\in\calF'}f_iz_i
    \tag{$\operatorname{M-IR}$}\label{lp:iterative}\\
    \text{s.t.\;} & z(U)=1\quad\forall U\in\calU\tag{$\operatorname{M-IR}$.1}\label{lp:iterative1}\\
    & z(S)\leq r_{\calM'}(S)\quad\forall S\subseteq\calF'\tag{$\operatorname{M-IR}$.2}\label{lp:iterative2}\\
    & z(\ib_j)\in [r-1,r]\quad\forall j\in D'\setminus(D_0\cup D_1)\tag{$\operatorname{M-IR}$.3}\label{lp:iterative3}\\
    & z_i\in[0,1]\quad\forall i\in\calF'.\tag{$\operatorname{M-IR}$.4}\label{lp:iterative4}
\end{alignat}

Before iterative rounding, $D_0\cup D_1$ is empty, therefore \ref{lp:iterative} only consists of the sum over $D'$ and fractional facility opening costs. 
We show its optimal objective is a good estimate of \ref{lp:natural:matroid}, \emph{excluding} the contributions of safe clients.

\begin{lemma}\label{lemma:lp:objective}
Before iterative rounding, \ref{lp:iterative} has optimum at most $\ftotal+\sum_{k\in D}r\dav(k)$.
\end{lemma}
\begin{proof}
Consider the initial solution $y$, and it satisfies all constraints by definition of \ref{lp:natural:matroid}. For each $j\in D'$, the contribution of $j$ in \ref{lp:natural:matroid} is the same as \ref{lp:iterative} for facilities in $\ib_j$, and at least $\rg\dma(j)$ for each facility in $\calF_j\setminus \ib_j$ since $\ib_j=\ball_{1/\gamma}(j)$, thus the contribution of $j$ to \ref{lp:iterative} is at most $\sum_{i\in\calF_j}y_id(i,j)=r\dav(j)$. For $k\in D\setminus D'$ marked by $j\in D'$, we have $\dav(k)\geq\dav(j)$. The lemma then follows by noticing $D_0\cup D_1=\emptyset$ in the beginning and taking the sum over $D$.
\end{proof}

The constraints in \ref{lp:iterative} are defined by two disjoint families $\calU,\,\calB$, and a matroid $\calM'$.
At a high level, \cref{algorithm:iterative} does the following.
When a constraint in \eqref{lp:iterative3} is tight with $z(\ib_j)=r$, we add $j$ to $D_1$ and this constraint is removed from \eqref{lp:iterative3}; 
when a constraint in \eqref{lp:iterative3} is tight with $z(\ib_j)=r-1$, we add $j$ to $D_0$ and this constraint is also removed;
when all constraints in \eqref{lp:iterative3} are removed, the resulting LP represents the intersection of a partition matroid $\calU$ (\eqref{lp:iterative1}) and the matroid $\calM'$ (\eqref{lp:iterative2}), thus has integral vertex solutions (see, e.g., \cite{edmonds2001submodular}); 
\cref{algorithm:iterative} outputs an integral optimal solution $z^\star$. 
We define $\hat F=\{i:i\in\calF',z^\star_i=1\}$ as the final solution, which concludes our algorithm.

\begin{algorithm}[hbt!]
\caption{\textsc{ALG-Iterative}}\label{algorithm:iterative}
\DontPrintSemicolon
\SetKwInOut{Input}{Input}
\SetKwInOut{Output}{Output}

\Input{$D',\calU,\calU_{shell}\subseteq\calU;\,\{\que_j\subseteq\calU:j\in\calC\}$}
\Output{an integral solution to \ref{lp:iterative}, a partition of $D'$}
$D_0\leftarrow\emptyset,\,D_1\leftarrow\emptyset$\;
\While{true}{
find an optimal vertex solution $z$ to \ref{lp:iterative}\;
\For($\calF'\leftarrow\calF'\setminus\{i\}$){every $i\in\calF'$ s.t. $z_i=0$}{
\lFor{$U\in\calU$}{$U\leftarrow U\setminus\{i\}$}
}
\uIf{there exists $j\in D'\setminus(D_0\cup D_1)$ s.t. $z(\ib_j)=r$}{
$D_1\leftarrow D_1\cup\{j\}$,
$U\leftarrow\ib_j\setminus Q_{j,<r}$, $Q_{j,r}\leftarrow U$, $\calU'\leftarrow\{U'\in\calU:U'\cap U\neq\emptyset\}$,
$\calU\leftarrow(\calU\setminus\calU')\cup\{U\}$\;\label{line:delete:shells}
\lFor{$j'\in D',\,Q_{j',r}\in\calU'$}
{$Q_{j',r}\leftarrow U$}}
\lElseIf{there exists $j\in D'\setminus(D_0\cup D_1)$ s.t. $z(\ib_j)=r-1$}{$D_0\leftarrow D_0\cup\{j\}$}
\lElse{break}
}
\Return $z^\star=z,\,D_0,\,D_1$
\end{algorithm}


\subsection{Analysis}

The following lemma gives a lower bound on the \emph{total loss} of the objective in \ref{lp:iterative} during \cref{algorithm:iterative}, since in each iteration we are also minimizing the objective by finding $z$.

\begin{lemma}\label{lemma:solution:transition}
After adding $j\in D'$ to $D_1$, \ref{lp:iterative} has the same objective w.r.t. the current $z$. 
After adding $j\in D'$ to $D_0$, \ref{lp:iterative} has the objective decreased by $\rg\dmd_j\dma(j)$ w.r.t. the current $z$.
\end{lemma}
\begin{proof}
Suppose we add $j$ to $D_1$, then we consider the following cases. 
For $j'\in D_0$, its contribution to \ref{lp:iterative} only depends on $Q_{j',<r}$, which is a subset of $\ib_{j'}$ and disjoint from $\ib_j$, using \cref{lemma:first:separation} and \cref{lemma:queue:and:bundle}. Therefore, none of the $(r-1)$ bundles in $\{Q_{j',1},\dots,Q_{j',r-1}\}$ is in $\calU'$ (the removed bundles) and the contribution is the same. 
For another $j'\neq j,\,j'\in D_1$, its contribution depends on $Q_{j',\leq r}$. 
The first $(r-1)$ bundles therein are inside $\ib_{j'}$ thus unaffected, and the last bundle $Q_{j',r}$ is also a subset of $\ib_{j'}$ hence not in $\calU'$, by considering the iteration in which we add $j'$ to $D_1$ and set $Q_{j',r}=\ib_{j'}\setminus Q_{j',<r}$ at an earlier time. 
For $j$ itself, we have $z(\ib_j)=r$ at the start of the iteration and $Q_{j,\leq r}=\ib_j$ at the end of the iteration by definition, thus the contribution of $j$ also does not change w.r.t. the current solution $z$. 
Finally, for $j'\neq j,\,j'\in D'\setminus(D_0\cup D_1)$, its contribution to the objective only depends on $z$ and $\ib_{j'}$, thus stays unchanged.

Suppose we add $j$ to $D_0$, then none of the bundles is removed, therefore the contribution of each $j'\in D'\setminus\{j\}$ does not change. 
We have $z(\ib_j)=r-1$ and thus $Q_{j,<r}=\ib_j$, since $Q_{j,<r}\subseteq\ib_j$ and $z(Q_{j,<r})=\sum_{t=1}^{r-1}z(Q_{j,t})=r-1$ according to \eqref{lp:iterative1}. The second assertion thus follows.
\end{proof}

\begin{lemma}\label{lemma:integral:after:rounding}
When \cref{algorithm:iterative} ends, the output solution $z^\star$ is integral and $D'=D_0\cup D_1$.
\end{lemma}
\begin{proof}
Assume the algorithm ends with $|D'\setminus(D_0\cup D_1)|\geq1$ but none of the corresponding constraints in \eqref{lp:iterative3} is tight. The other tight constraints then belong to the intersection polytope of $\calU$ (a partition matroid) and $\calM'$, thus the output solution is integral  (see, e.g., \cite{schrijver2003combinatorial}).
This in turn forces every constraint left in \eqref{lp:iterative3} to be tight, which is a contradiction.
\end{proof}

\begin{figure}[t]
    \centering
    \includegraphics[width=0.8\textwidth]{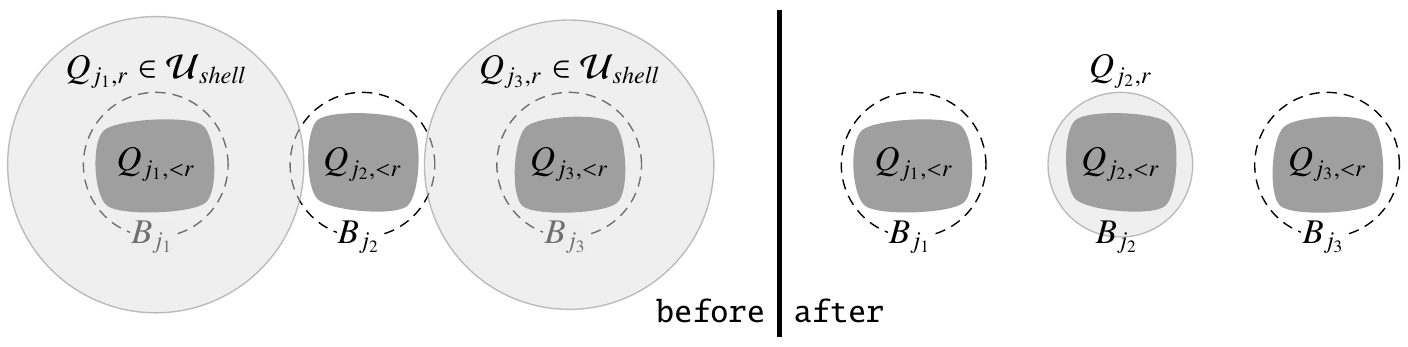}
    \caption{An illustration when $z(\ib_{j_2})=r$ becomes tight at \eqref{lp:iterative3}, and shell bundles $Q_{j_1,r},Q_{j_3,r}$ are removed due to $\calU'$ at \cref{line:delete:shells}. This also creates a new bundle $Q_{j_2,r}\subseteq\ib_{j_2}$, such that $\ib_{j_2}=Q_{j_2,<r}\cup Q_{j_2,r}=Q_{j_2,\leq r}$.
    }
    \label{figure:shell:bundle}
\end{figure}

By \cref{lemma:first:separation} and \cref{lemma:queue:and:bundle}, any non-shell bundle created by a dangerous $j\in D'$ is inside $B_j$ and disjoint from other closed balls in $\calB$; 
any bundle created by a safe client is either in $Q_{j,<r}\subseteq B_j$ or disjoint from $\ib_j$ for each $j\in D'$.
Thus, for each $U'\in\calU'$ at \cref{line:delete:shells}, it is not hard to see that $U'\in\calU_{shell}$.
See \cref{figure:shell:bundle} for an illustration. 
Using this crucial observation, we present the following lemmas.

\begin{lemma}\label{lemma:final:solution:dangerous}
For each $j\in D'$, after \cref{algorithm:iterative}, there are at least $(r-1)$ bundles in $\ib_j$, and the $r$-th nearest bundle is within distance at most $\frac{3\gamma^2-\gamma+2}{\gamma(\gamma-1)}\dma(j)$ from $j$.
\end{lemma}
\begin{proof}
The initial $(r-1)$ bundles in $Q_{j,<r}$ remain in $\calU$ since they are not shell bundles. For the $r$-th nearest bundle, since $D_0\cup D_1=D'$ according to \cref{lemma:integral:after:rounding}, we consider $D_1,D_0$ separately. 
For $j\in D_1$, $Q_{j,r}\subseteq\ib_j$ by the definition of \cref{algorithm:iterative}, hence $\dma(j,Q_{j,r})\leq\dma(j)/\gamma$. 
For $j\in D_0$, we consider the original bundle $Q_{j,r}$ before iterative rounding. If $Q_{j,r}\notin\calU_{shell}$, it is never removed and has $\dma(j,Q_{j,r})\leq3\dma(j)$ by \cref{lemma:bundle:third:away}, otherwise we consider the following cases.

If $Q_{j,r}\in\calU_{shell}$ is created by $j$, it means $Q_{j,r}$ is created from the candidate bundle proposed by $j$, thus $Q_{j,r}=\calF_j\setminus Q_{j,<r}$ and $\dma(j,Q_{j,r})=\dma(j)$ using \cref{lemma:queue:and:bundle}.
If it is not removed, it satisfies $\dma(j,Q_{j,r})=\dma(j)$; 
otherwise, it is removed when we move another $j'\in D'$ to $D_1$ with a new bundle $Q_{j',r}\subseteq \ib_{j'}$, thus $d(j,j')\leq\dma(j)+\dma(j')/\gamma$ by the triangle inequality. 
Using \cref{lemma:first:separation}, one has $d(j,j')\geq\dma(j')-\dma(j)/\gamma$, and thus $\dma(j')\leq\frac{\gamma+1}{\gamma-1}\dma(j)$. 
Using the triangle inequality, $Q_{j',r}$ (which persists to the end) satisfies $\dma(j,Q_{j',r})\leq\dma(j)+2\dma(j')/\gamma\leq \left(1+\frac{2\gamma+2}{\gamma(\gamma-1)}\right)\dma(j)$.

If $Q_{j,r}\in\calU_{shell}$ is created \emph{earlier} by $j'\neq j$ when $|\que_j|<r$, $\dma(j')=\dma(j',Q_{j,r})\leq\dma(j)$. 
If $Q_{j,r}$ is not removed, one has $\dma(j,Q_{j,r})\leq3\dma(j)$ using \cref{lemma:bundle:third:away};
otherwise, when we move $j''$ to $D_1$ with a new bundle $Q_{j'',r}\subseteq \ib_{j''}$ and remove $Q_{j,r}$, we likewise have $\dma(j'')\leq\frac{\gamma+1}{\gamma-1}\dma(j')$. 
Since $\dma(j')\leq\dma(j)$, we have $\dma(j,Q_{j'',r})\leq 3\dma(j)+2\dma(j'')/\gamma\leq\left(3+\frac{2\gamma+2}{\gamma(\gamma-1)}\right)\dma(j)$ by the triangle inequality.
\end{proof}

\begin{lemma}\label{lemma:final:solution:safe}
For each $j\in \calC\setminus D$ with an initial $|\que_j|=l\leq r$, after \cref{algorithm:iterative}, there are $l$ bundles at distances at most $3\dma^t(j)$ for $t\in[l]$, respectively, and $(r-l)$ more bundles within distances $3\dma^{l+1}(j),\dots,3\dma^{r-1}(j)$, $\frac{7\gamma^2-2\gamma+3}{(\gamma-1)^2}\dma^{r}(j)$, respectively.
\end{lemma}
\begin{proof}
The $l$ bundles in $\que_j$ are never removed since they are not shell bundles, yielding the first assertion by \cref{lemma:bundle:third:away}. We have the following cases when $l<r$.

If $\calF_j'\leftarrow\emptyset$ by $\ib_{j'},j'\in D'$ when $j$ proposes $U$ at \cref{line:noalien}, by \cref{lemma:noalien:safe}, one has $\dma(j,U)\geq(1-\rg)\dma(j')/2$ and thus $d(j,j')\leq\dma(j,U)+\dma(j')/\gamma\leq \frac{\gamma+1}{\gamma-1}\dma(j,U)$. 
There are at least $(r-1)$ bundles in $\ib_{j'}$, each at a distance of $\dma(j,U)+2\dma(j')/\gamma\leq\left(1+\frac{4}{\gamma-1}\right)\dma(j,U)< 3\dma(j,U)$ away from $j$.
By \cref{lemma:final:solution:dangerous}, there is also a distinct $r$-th bundle that is at most $\frac{3\gamma^2-\gamma+2}{\gamma(\gamma-1)}\dma(j')$ away from $j'$, thus at most $d(j,j')+\frac{3\gamma^2-\gamma+2}{\gamma(\gamma-1)}\frac{2}{1-\rg}\dma(j,U)\leq\frac{7\gamma^2-2\gamma+3}{(\gamma-1)^2}\dma(j,U)$ away from $j$ using the triangle inequality.

If $\calF_j'\leftarrow\emptyset$ by $U'\in\calU_{shell}$ when $j$ proposes $U$ at \cref{line:noshell}, suppose $U'$ is created earlier by a dangerous $j'\in D'$. 
We have $\dma(j')=\dma(j',U')\leq\dma(j,U)$ and $d(j,j')\leq \dma(j,U)+\dma(j',U')\leq 2\dma(j,U)$.
Following a similar argument as before, there are at least $(r-1)$ bundles in $\ib_{j'}$, each at a distance of at most $3\dma(j,U)$, 
and a distinct $r$-th bundle within distance $\frac{5\gamma^2-3\gamma+2}{\gamma(\gamma-1)}\dma(j,U)$ from $j$.

In both cases, among the $r$ distinct bundles associated with $j'$, we choose $(r-l)$ bundles that are not in $\que_j$. 
Since $U$ is the candidate bundle when $|\que_j|=l$ and at most $l$ volume of facilities is removed from $\calF_j'$, one has $\dma(j,U)\leq\dma^{l'}(j)$ for $l'>l$, thus the second assertion follows.
\end{proof}

\begin{theorem}\label{theorem:matroid:basic}
There exists a $(138+O(\delta))$-approximation for \emph{\ftm} in polynomial time independent of $\delta>0$.
\end{theorem}
\begin{proof}
By \cref{lemma:lp:objective}, \cref{lemma:solution:transition} and \cref{lemma:integral:after:rounding}, $z^\star\in\{0,1\}^{\calF'}$ has an objective at most $\ftotal+\sum_{k\in D}r\dav(k)-\rg\sum_{j\in D_0}\dmd_j\cdot\dma(j)$ in \ref{lp:iterative}. 
We let $\delta>0$ be small enough and recall $\gamma=3+\delta$.
For $j\in D_1$, the cost of assigning $r$ open facilities in $\hat F$ is equal to its contribution in \ref{lp:iterative}, since there are $r$ bundles in its objective and thus $r$ open facilities due to \eqref{lp:iterative1}. 
For $j\in D_0$, the cost of assigning the nearest $(r-1)$ open facilities is equal to its contribution in \ref{lp:iterative}, and the costs of the $r$-th assignments, summed over $D_0$, is at most $\frac{3\gamma^2-\gamma+2}{\gamma(\gamma-1)}\sum_{j\in D_0}\dmd_j\cdot\dma(j)=\frac{13+O(\delta)}{\gamma}\sum_{j\in D_0}\dmd_j\cdot\dma(j)$ due to \cref{lemma:final:solution:dangerous}.

Next, for a safe client $j\in\calC\setminus D$, the cost of assigning $r$ open facilities is no larger than that of the open facilities in the $r$ bundles described in \cref{lemma:final:solution:safe}, thus at most,
\begin{align}
    &\sum_{t=1}^{r-1}3\dma^t(j)+\frac{7\gamma^2-2\gamma+3}{(\gamma-1)^2}\dma^r(j)\notag\\
    &\leq3\sum_{t=2}^r\dav^t(j)+\frac{21\gamma^3-6\gamma^2+9\gamma}{(\gamma-1)^2}\dav^r(j)\notag\\
    &\leq(138+O(\delta))r\dav(j).
\label{eqn:matroid:safe:objective}
\end{align}

To compensate for the dangerous clients in $D\setminus D'$ that are ``moved to'' $D'$, an additional cost of at most $6\sum_{k\in D}r\dav(k)$ is incurred, since for $k\in D\setminus D'$ marked by $j\in D'$ in the filtering process, we have $d(j,k)\leq6\dav(k)$, and relocating the demand from $j$ to $k$ increases the total cost by at most $6r\dav(k)$ using the triangle inequality. 
To summarize, the total cost is at most
\begin{align}
&\ftotal+\sum_{k\in D}r\dav(k)+\frac{12+O(\delta)}{\gamma}\sum_{j\in D_0}\dmd_j\cdot\dma(j)\notag\\
&+(138+O(\delta))\sum_{j\in\calC\setminus D}r\dav(j)+6\sum_{k\in D}r\dav(k)\notag\\
&\leq(138+O(\delta))\Big(\ftotal+\sum_{j\in\calC}r\dav(j)\Big),
\label{eqn:matroid:final:objective}
\end{align}
where $\rg\sum_{j\in D_0}\dmd_j\cdot\dma(j)$ is at most the total decrease of objective in \cref{algorithm:iterative}, hence bounded by $\ftotal+\sum_{k\in D}r\dav(k)$ using \cref{lemma:lp:objective} and \cref{lemma:solution:transition}. 
The sum in \cref{eqn:matroid:final:objective} is the objective of \ref{lp:natural:matroid}, which is at most the optimum of the original problem, thus the theorem follows.
\end{proof}

\begin{remark}
Let us briefly review the arguments for \cref{eqn:matroid:final:objective}, which are also crucial in the knapsack version:
For each $j\in D'$, the assignment cost is obtained via the $r$ or $(r-1)$ bundles that constitute $B_j$, and an additional cost is incurred using \cref{lemma:final:solution:dangerous} if $j\in D_0$; 
for safe clients, the assignment costs are obtained using \cref{lemma:final:solution:safe} and \cref{eqn:matroid:safe:objective}, where for each safe $j$ and bundle $U$, we only care about $\dma(j,U)$.
These observations are helpful for reusing the arguments and lemmas in the next section.
\end{remark}

\section{Fault-tolerant knapsack median}\label{section:knap:uniform}

We consider \ftk. 
Fix an unknown optimal solution with objective $\opt\geq0$ and total facility opening cost $\opt_f\geq0$. 
We strengthen the natural relaxation using Kumar's guessing method \cite{kumar2012constant} as follows.
Let $\delta_j\geq0$ be the distance from $j$ to its $r$-th nearest open facility in the optimal solution. 
For each $j'$, the distance from $j'$ to its $r$-th nearest open facility is at least $\max\{0,\delta_j-d(j,j')\}$ using the triangle inequality, hence
\begin{equation}
\sum_{j'\in\calC}\max\{0,\delta_j-d(j,j')\}\leq\opt,\,\forall j\in\calC.\label{equation:kumar:knapsack}
\end{equation}

Fix a small $\epsilon>0$. 
We first guess $\opt$ and $\opt_f$ up to a multiplicative factor of $(1+\epsilon)$, via exhaustive search on integer powers of $(1+\epsilon)$.
In the rest of this section, assume we have made the correct guesses such that $\opt'\in[\opt,(1+\epsilon)\opt]$ and $\opt_f'\in[\opt_f,(1+\epsilon)\opt_f]$.
For each $j\in\calC$, we replace $\opt$ with $\opt'$ in \cref{equation:kumar:knapsack} and choose $\Delta_j=\max\{\delta_j\geq0:\text{\cref{equation:kumar:knapsack} holds true}\}$. 
This ensures that $\Delta_j\geq\delta_j$ for each $j\in\calC$, thus we can add new constraints $x_{ij}=0,\,\forall d(i,j)>\Delta_j$ and $y_i=0,\,\forall f_i>\opt_f'$, and the LP is still satisfied by the optimal solution.
The relaxation is defined as follows.
\begin{alignat*}{1}
    \text{min\quad} &
    \sum_{i\in\calF}f_iy_i+\sum_{j\in\calC}\sum_{i\in\calF}x_{ij}d(i,j)\tag{$\operatorname{K-LP}$}\label{lp:natural:knapsack}\\
    \text{s.t.\quad}
    & \sum_{i\in\calF}x_{ij}= r \quad\forall j\in\calC;\qquad \sum_{i\in\calC}w_iy_i \leq W\\
    & 0\leq x_{ij} \leq y_i\leq1\quad \forall i\in\calF,j\in\calC\\
    & x_{ij} =0 \quad\forall i\in\calF,j\in\calC\mathrm{\;s.t.\,}d(i,j) > \Delta_j
    \tag{$\operatorname{K-LP}$.3}\label{lp:natural:knapsack4}\\
    & y_i =0 \quad\forall i\in\calF\mathrm{\;s.t.\,}f_i > \opt_f'.
    \tag{$\operatorname{K-LP}$.4}\label{lp:natural:knapsack5}
\end{alignat*}

Let $(x,y)$ be an optimal solution to \ref{lp:natural:knapsack} in what follows.
We define the two disjoint families $\calB$ and $\calU$ in the same way as \ftm.
In the following auxiliary LP, we still use $D_0$ and $D_1$ to partition $D'$.
Unlike the matroid case, we need extra constraints to make sure each facility is opened at most once.
Let $g:\calF'\rightarrow\calF$ take any duplicate to the original facility; $g^{-1}(i)\subseteq\calF'$ is the collection of copies of $i\in\calF$.
\eqref{lp:iterative:knap2} says that at most one copy of each facility is opened, similar to \cite{hajiaghayi2016constant}.

\begin{alignat*}{1}
    \text{min\;}& \sum_{j\in D'\setminus(D_0\cup D_1)} \dmd_j\left(\left(\sum_{i\in \ib_j}z_i d(i,j)\right) +\frac{\dma(j)}{\gamma}(r-z(\ib_j)) \right)+\\
    &\sum_{j\in D_0}\dmd_j\sum_{i\in Q_{j,<r}}z_i d(i,j)+
    \sum_{j\in D_1}\dmd_j\sum_{i\in Q_{j,\leq r}}z_i d(i,j)+\sum_{i\in\calF'}f_iz_i
    \tag{$\operatorname{K-IR}$}\label{lp:iterative:knap}\\
    \text{s.t.\;} & z(U)=1\quad\forall U\in\calU\tag{$\operatorname{K-IR}$.1}\label{lp:iterative:knap1}\\
    & z(g^{-1}(i))\leq1\quad\forall i\in\calF\tag{$\operatorname{K-IR}$.2}\label{lp:iterative:knap2}\\
    & z(\ib_j)\in [r-1,r]\quad\forall j\in D'\setminus(D_0\cup D_1)\tag{$\operatorname{K-IR}$.3}\label{lp:iterative:knap3}\\
    & \sum_{i\in\calF'}w_iz_i\leq W\tag{$\operatorname{K-IR}$.4}\label{lp:iterative:knap4}\\
    & z_i\in[0,1]\quad\forall i\in\calF'\tag{$\operatorname{K-IR}$.5}\label{lp:iterative:knap5}\\
    & z_i=0\quad\forall i\in\calF',\,f_{g(i)}>\opt'_f.\tag{$\operatorname{K-IR}$.6}\label{lp:iterative:knap6}
\end{alignat*}

The following lemma for \ref{lp:iterative:knap} is analogous to \cref{lemma:lp:objective}.
\begin{lemma}\label{lemma:lp:objective:knap}
Before iterative rounding, \ref{lp:iterative:knap} has optimum at most $\ftotal+\sum_{k\in D}r\dav(k)$.
\end{lemma}

We use the same \cref{algorithm:iterative} and obtain the output $z^\star$. 
Akin to \cite{krishnaswamy2018constant}, $z^\star$ is not necessarily integral because of the knapsack constraint \eqref{lp:iterative:knap4}. 
In the rest of this section, we obtain an integral solution $\hat z$ from $z^\star$ and show the following theorem.

\begin{theorem}\label{theorem:knapsack:basic}
There exists a $(143.34+O(\epsilon))$-approximation for \emph{\ftk}. 
The running time of the algorithm is a polynomial of the input size and $1/\epsilon$.
\end{theorem}

\subsection{Obtaining an integral solution}

For $i\in\calF$ s.t. $z^\star(g^{-1}(i))\in(0,1)$, we say $i$ is \emph{non-tight}; 
we say $i\in\calF$ is \emph{upper-tight} if $z^\star(g^{-1}(i))=1$, and \emph{lower-tight} if $z^\star(g^{-1}(i))=0$. 
Suppose \cref{algorithm:iterative} returns $z^\star$ with \emph{none} of the remaining constraints in \eqref{lp:iterative:knap3} tight, $l$ strictly fractional variables in $z^\star$ and $T\geq0$ non-tight facilities. 
It is well-known that (see, e.g., \cite{schrijver2003combinatorial}) there exist $l$ tight and linearly-independent (of the trivial tight constraints $z_i\in[0,1]$) constraints in \eqref{lp:iterative:knap1}, \eqref{lp:iterative:knap2} and \eqref{lp:iterative:knap4}. 
Since the bundles are disjoint, at most $\floor{l/2}$ constraints in \eqref{lp:iterative:knap1} are tight on these $l$ fractions;
because there are $T$ non-tight facilities, at most $\floor{(l-T)/2}$ constraints in \eqref{lp:iterative:knap2} are tight on these $l$ fractions. 
We then have $\floor{l/2}+\floor{(l-T)/2}+1\geq l$, and thus $T\in\{0,1,2\}$.


We study the three cases of $T\in\{0,1,2\}$ in the following and apply different rounding algorithms. 
The approximation guarantee in \cref{theorem:knapsack:basic} is the maximum over these cases.

\subparagraph{When $T=1$.}
Because $\floor{l/2}+\floor{(l-1)/2}+1=l$, we have exactly $\floor{l/2}$ tight constraints in \eqref{lp:iterative:knap1}, $\floor{(l-1)/2}$ in \eqref{lp:iterative:knap2} and one in \eqref{lp:iterative:knap4}. 

Since there is only one non-tight facility denoted by $i_0\in\calF'$, and these $l$ constraints are linearly independent, it is not hard to see that $l$ must be an odd number $(2k+1)$, and these $(2k+1)$ facilities can be renamed $(i_0,i_1,i_2,\dots,i_{2k-1},i_{2k})$ such that:
\begin{itemize}
    \item $\{i_0,i_1\},\dots,\{i_{2k-2},i_{2k-1}\}$ are $k$ tight bundles; 
    i.e., one has $z^\star_{i_0}+z^\star_{i_1}=\cdots=z^\star_{i_{2k-2}}+z^\star_{i_{2k-1}}=1$.
    \item $\{i_{1},i_{2}\},\dots,\{i_{2k-1},i_{2k}\}$ are copies of $k$ upper-tight facilities;
    i.e., each pair is co-located, $z^\star_{i_{1}}+z^\star_{i_{2}}=\cdots=z^\star_{i_{2k-1}}+z^\star_{i_{2k}}=1$.
\end{itemize}

To obtain an integral solution, we let $\hat z_{i_1}=\cdots=\hat z_{i_{2k-1}}=1$ and $\hat z_{i_0}=\cdots=\hat z_{i_{2k}}=0$. 
For each $i\notin\{i_0,\dots,i_{2k}\}$, set $\hat z_i=z^\star_i$. 
Replace the bundles $\{i_0,i_1\},\dots,\{i_{2k-2},i_{2k-1}\}$ with their subsets $\{i_1\},\dots,\{i_{2k-1}\}$, respectively.

Our final solution is then $\hat F=\{g(i):i\in\calF',\hat z_i=1\}$, and $\sum_{i\in \hat F}w_i\leq\sum_iw_iz^\star_i\leq W$ follows since we reduce the total weight by closing $i_0$. 
Therefore, $\hat F$ is feasible.
Similarly, we also have $\sum_{i\in \hat F}f_i\leq\sum_if_iz^\star_i$. 

Next, we focus on the cost of assigning facilities to clients.
Since each closed ball $B_j\in\calB$ is defined solely based on the underlying metric, for an original $i\in\calF$, either $g^{-1}(i)\subseteq B_j$ or $g^{-1}(i)\cap B_j=\emptyset$; 
we say that $B_j\subseteq\calF'$ is \emph{copy-consistent}.

If $\exists j_0\in D'$ s.t. $i_0\in\ib_{j_0}$, because the closed balls are copy-consistent and disjoint, we have $z^\star(B_{j'})\in\Z,\,\forall j'\in D'\setminus\{j_0\}$ and $D_0\cup D_1=D'\setminus\{j_0\}$ by \cref{algorithm:iterative}.
To assign $r$ open facilities to all clients, we have the following analysis.
\begin{enumerate}
    \item For $D_0\cup D_1$ and $\calC\setminus D$. 
    First, we have only changed some bundles to their subsets and there is still exactly one open facility in each bundle.
    Second, it is clear that by refining $B_j=Q_{j,<r}$ for $j\in D_0$ and $B_j=Q_{j,\leq r}$ for $j\in D_1$ using \cref{algorithm:iterative} (see the analysis in \cref{lemma:solution:transition}), the scopes of sums for $D_0\cup D_1$ in \ref{lp:iterative:knap} have also become copy-consistent.
    Since $i_0\notin B_j$ for each $j\in D_0\cup D_1$, $z^\star$ and $\hat z$ incur the same contributions.
    Thus, the same arguments for these clients in \cref{lemma:final:solution:dangerous}, \cref{lemma:final:solution:safe} and \cref{theorem:matroid:basic} are still valid.
    
    \item For $j_0$, there are $(r-1)$ open facilities in $\ib_{j_0}$, due to \eqref{lp:iterative:knap3} and that we close $i_0$ as described above. 
    Similar to \cref{lemma:final:solution:dangerous}, $j_0$ needs to pay at most $\frac{3\gamma^2-\gamma+2}{\gamma(\gamma-1)}\cdot\dmd_{j_0}\dma(j_0)$ for the $r$-th assignment. 
    However, because $j_0$ is \emph{not} in $D_0$, we cannot use \cref{eqn:matroid:final:objective} and have to bound the quantity above using another method.
    For each $j\in D$ marked by $j_0$, we have $d(j,j_0)\leq6\dav(j)$ by definition. 
    Using \eqref{lp:natural:knapsack4} and \cref{equation:kumar:knapsack}, one has $\dma(j_0)\leq \Delta_{j_0}$ and        \begin{equation}
         \dmd_{j_0}\dma(j_0)=\sum_{\substack{j\text{ marked}\\ \text{by }j_0}}\dma(j_0)
        \leq\opt'+\sum_{\substack{j\text{ marked}\\
        \text{by }j_0}}6\dav(j).\label{equation:kumar:application}
        \end{equation}
\end{enumerate}

Using \cref{equation:kumar:application} and the same analysis as \cref{eqn:matroid:final:objective}, the total cost of $\hat F$ is at most
$\frac{3\gamma^2-\gamma+2}{\gamma(\gamma-1)}\opt'\leq(13/3+O(\epsilon+\delta))\opt$ plus $(138+O(\delta))\cdot\opt$, thus the approximation ratio is $(142.34+O(\epsilon+\delta))$.
    
On the other hand, if $\forall j_0\in D',\,i_0\notin\ib_{j_0}$, we have $D'=D_0\cup D_1$ since no such $\ib_j$ contains non-tight facilities. 
Again, we only change some bundles into their subsets during post-processing, thus the analysis in \cref{theorem:matroid:basic} still holds with an approximation ratio of $(138+O(\delta))$.

\subparagraph{When $T=2$.}

To satisfy $\floor{l/2}+\floor{(l-2)/2}+1\geq l$, we must have $l=2k$ as an even number.
Because there are two non-tight facilities denoted by 
$i_0,i_{2k-1}\in\calF'$,
these $2k$ facilities can be similarly renamed $(i_0,i_1,i_2,\dots,i_{2k-1})$ such that:
\begin{itemize}
    \item $\{i_0,i_1\},\dots,\{i_{2k-2},i_{2k-1}\}$ are $k$ tight bundles; 
    i.e., one has $z^\star_{i_0}+z^\star_{i_1}=\cdots=z^\star_{i_{2k-2}}+z^\star_{i_{2k-1}}=1$.
    \item $\{i_{1},i_{2}\},\dots,\{i_{2k-3},i_{2k-2}\}$ are copies of $(k-1)$ upper-tight facilities;
    i.e., each pair is co-located, $z^\star_{i_{1}}+z^\star_{i_{2}}=\cdots=z^\star_{i_{2k-3}}+z^\star_{i_{2k-2}}=1$.
\end{itemize}

We assume w.l.o.g. that $w_{i_0}\geq w_{i_{2k-1}}$.
To obtain an integral solution, let $\hat z_{i_1}=\cdots=\hat z_{i_{2k-1}}=1$ and $\hat z_{i_0}=\cdots=\hat z_{i_{2k-2}}=0$.
Set $\hat z_i=z^\star_i$ for each $i\notin\{i_0,\dots,i_{2k-1}\}$.
Replace the bundles $\{i_0,i_1\},\dots,\{i_{2k-2},i_{2k-1}\}$ with their subsets $\{i_1\},\dots,\{i_{2k-1}\}$.
Our final solution is $\hat F=\{g(i):i\in\calF',\hat z_i=1\}$, and $\sum_{i\in\hat F}w_i\leq\sum_iw_iz^\star_i\leq W$ follows since $z^\star_{i_0}+z^\star_{i_{2k-1}}=1$ from the above and $\hat z_{i_0}w_{i_0}+\hat z_{i_{2k-1}}w_{i_{2k-1}}\leq z^\star_{i_0}w_{i_0}+z^\star_{i_{2k-1}}w_{i_{2k-1}}$ using $w_{i_0}\geq w_{i_{2k-1}}$.
For the total opening cost, we have $\sum_{i\in \hat F}f_i\leq f_{g(i_{2k-1})}+\sum_if_iz^\star_i\leq\opt_f'+\sum_if_iz^\star_i$ due to \eqref{lp:iterative:knap6}. 

For the cost of assigning facilities,
let us first assume that $i_0$ and $i_{2k-1}$ do not belong to the same closed ball in $\calB$ (if there are any).
The analysis in this case is almost the same as the case when $T=1$, except that an additional facility opening cost $\opt_f'\leq(1+\epsilon)\opt_f\leq(1+\epsilon)\opt$ increases the approximation ratio to $(143.34+O(\epsilon+\delta))$.

Then, if $i_0$ and $i_{2k-1}$ indeed belong to the same $B_{j_0}\in\calB$, since $z^\star_{i_0}+z^\star_{i_{2k-1}}=1$, we have $j_0\in D_0\cup D_1$.
If $j_0\in D_1$, by changing from $z^\star$ to $\hat z$, the extra contribution of $j_0$ to \ref{lp:iterative:knap} can be similarly bounded using \cref{equation:kumar:application};
if $j_0\in D_0$, we additionally use \cref{lemma:solution:transition} and \cref{lemma:final:solution:dangerous} on $j_0$ for its $r$-th assignment cost similar to \cref{eqn:matroid:final:objective}.
In both cases, the approximation ratio is at most $(143.34+O(\epsilon+\delta))$.

\subparagraph{When $T=0$.}

Every $i\in\calF$ is either upper-tight or lower-tight, and it is easy to see that $D'=D_0\cup D_1$. 
Let $\hat F=\{i\in\calF:z^\star(g^{-1}(i))=1\}$, and it is clear that $\sum_{i\in\hat F}w_i\leq W$ since $z^\star$ is feasible to \ref{lp:iterative:knap}.
To show that $\hat F$ is a good approximate solution, we first reconstruct some bundles.
Let $\calF_1\subseteq\calF'$ be the set of fractional facilities, and $\calU_1\subseteq\calU$ be the set of bundles supported on $\calF_1$;
by \eqref{lp:iterative:knap1}, it is clear that all other bundles are supported on integral variables in $z^\star$.
Each facility in $g(\calF_1)$ is upper-tight, hence $|g(\calF_1)|\geq|\calU_1|$.
We construct a network flow instance with source $s$ and sink $t$ as follows.

\begin{itemize}
    \item In the ``original'' layer, put $o_i$ for each $i\in g(\calF_1)$. Connect them to $s$ using links with unit capacity.
    \item In the ``copy'' layer, put $u_{i'}$ for each $i'\in\calF_1$. 
    Connect $u_{i'}$ to $o_{g(i')}$ using a link with unit capacity.
    \item In the ``bundle'' layer, put $v_U$ for each $U\in\calU_1$. 
    Also add a dummy node $v_\perp$ to this layer. 
    Connect each $v_U$ to $t$ using a link with unit capacity, and $v_\perp$ to $t$ using a link with capacity $(|g(\calF_1)|-|\calU_1|)$.
    \item Connect $(u_{i'},v_U)$ using a link with unit capacity if and only if $i'\in U$. Also connect $(u_{i'},v_\perp)$ using a link with unit capacity for each $i'\in\calF_1$.
\end{itemize}

We obtain an integral solution via an integral flow on the network.
Since members of $\calU_1$ are disjoint and supported on $\calF_1$, and every member of $g(\calF_1)$ is upper-tight, $z^\star$ naturally induces a fractional flow $\tilde f$ having flow value $z^\star(\calF_1)=|g(\calF_1)|$.
Because the capacities are all integers, we can round $\tilde f$ to an arbitrary integral flow $\hat f$ with the same flow value (see, e.g., \cite{schrijver2003combinatorial}). 
It is clear that its flow on each bundle node $v_U$ is 1, and its flow on $v_\perp$ is $(|g(\calF_1)|-|\calU_1|)$.
Replace $U\in\calU_1$ with $\{i'\}\subseteq\calF_1$ if and only if $\hat f(u_{i'},v_U)=1$, and $i'$ was previously in $U$ by the definition of the network; 
every bundle in $\calU_1$ becomes its own subset.
For $i'\notin\calF_1$, we open $g(i')$ if and only if $z^\star_{i'}=1$; we also open each $i\in g(\calF_1)$.
Easy to see that this solution is the same as $\hat F$.

Because we open each upper-tight $i\in\calF$ w.r.t. $z^\star$ and only change some bundles to their corresponding subsets, the same analysis in \cref{theorem:matroid:basic} holds and the approximation ratio in this case is $(138+O(\delta))$.

\section*{Acknowledgements}

The author thanks Jian Li for some helpful discussions.
The author thanks the editors and the anonymous referees at \textit{Operations Research Letters} for their constructive comments; in particular, the author would like to thank the late Gerhard Woeginger, who was an area editor at \textit{Operations Research Letters}, for providing valuable comments on improving the presentation of this paper.

\bibliographystyle{plainurl}
\bibliography{references.bib}

\end{document}